\documentclass[11pt,fleqn,letterpaper]{article}
\usepackage{fullpage,amsmath,amssymb,amsthm}
\usepackage{xspace}
\usepackage{multirow}
\usepackage[ruled,vlined,linesnumbered]{algorithm2e}
\usepackage[hypertex]{hyperref}

\def\whp{with high probability\xspace}

\newcommand{\REMOVE}[1]{}

\newcommand{\veps}{\varepsilon}

\newcommand{\tuple}[1]{{\langle{#1}\rangle}}
\newcommand{\minn}[1]{\min\{#1\}}

\newcommand{\aset}[1]{\{#1\}}

\DeclareMathOperator{\supp}{supp}
\DeclareMathOperator{\diam}{diam}

\DeclareMathOperator{\EX}{\mathbb E}
\DeclareMathOperator{\LP}{\mathrm{LP}}

\newtheorem{theorem}{Theorem}[section]
\newtheorem{lemma}[theorem]{Lemma}

\newtheorem{corollary}[theorem]{Corollary}

\theoremstyle{definition}
\newtheorem{definition}[theorem]{Definition}

\def\compactify{\itemsep=0pt \topsep=0pt \partopsep=0pt \parsep=0pt}
\def\Lovasz{Lov\'asz\xspace}

\newcounter{this-list}

\newcounter{par-list}
\newlength{\parlistlength}
\begin{document}

%\pagestyle{myheadings}
%\addtolength{\headsep}{0.3in}

\title{Fault-Tolerant Spanners: Better and Simpler}
\author{Michael Dinitz\thanks{Email: \href{mailto:michael.dinitz@weizmann.ac.il}{\texttt{michael.dinitz@weizmann.ac.il}}} \qquad\qquad
Robert Krauthgamer\thanks{Supported in part by The Israel Science Foundation (grant \#452/08), and by a Minerva grant.
Email: \href{mailto:robert.krauthgamer@weizmann.ac.il}{\texttt{robert.krauthgamer@weizmann.ac.il}}
}
\\
Weizmann Institute of Science
}

\begin{titlepage}
\maketitle
\thispagestyle{empty}

\begin{abstract}
A natural requirement of many distributed structures is \emph{fault-tolerance}: after some failures, whatever remains from the structure should still be effective for whatever remains from the network.  In this paper we examine spanners of general graphs that are tolerant to vertex failures, and significantly improve their dependence on the number of faults $r$, for all stretch bounds.

For stretch $k \geq 3$ we design a simple transformation that converts \emph{every} $k$-spanner construction with at most $f(n)$ edges into an $r$-fault-tolerant $k$-spanner construction with at most  $O(r^3 \log n) \cdot f(2n/r)$ edges.  Applying this to standard greedy spanner constructions gives $r$-fault tolerant $k$-spanners with $\tilde O(r^{2} n^{1+\frac{2}{k+1}})$ edges.  The previous construction by Chechik, Langberg, Peleg, and Roddity [STOC 2009] depends similarly on $n$ but \emph{exponentially} on $r$ (approximately like $k^r$).

For the case $k=2$ and unit-length edges, an $O(r \log n)$-approximation algorithm is known from recent work of Dinitz and Krauthgamer [arXiv 2010], where several spanner results are obtained using a common approach of rounding a natural flow-based linear programming relaxation.  Here we use a different (stronger) LP relaxation and improve the approximation ratio to $O(\log n)$, which is, notably, \emph{independent} of the number of faults $r$.  We further strengthen this bound in terms of the maximum degree by using the \Lovasz Local Lemma.

Finally, we show that most of our constructions are inherently local by designing equivalent distributed algorithms in the $\mathcal{LOCAL}$ model of distributed computation.
\end{abstract}
\end{titlepage}

\section{Introduction}

Let $G=(V,E)$ be a graph, possibly with edge-lengths $\ell : E \rightarrow \mathbb{R}_{\geq 0}$.
A \emph{$k$-spanner} of $G$, for $k\ge 1$, is a subgraph $G'=(V,E')$ that preserves all pairwise distances within factor $k$, i.e.~for all $u,v\in V$,
\begin{equation} \label{eq:defn}
  d_{G'}(u,v) \leq k\cdot d_G(u,v).
\end{equation}
Here and throughout, $d_H$ denotes the shortest-path distance in a graph $H$,
and $n=|V|$.  The distance preservation factor $k$ is called the \emph{stretch} of the spanner.  It is easy to see that requiring \eqref{eq:defn} only for edges $(u,v)\in E$ suffices.  This definition also extends naturally to \emph{directed} graphs.  Obviously $G$ is a $1$-spanner of itself, so usually the goal is to compute a ``small'' spanner.  Two traditional notions of ``small'' are the number of edges in $G'$ (called the \emph{size} of $G'$), and the \emph{weight} of $G'$ (where the weight of a graph is the sum of the lengths of the edges in the graph).  If every edge has unit length then these two notions are the same, but for more general edge lengths they can be quite different.

This notion of graph spanners, first introduced by Peleg and Sch{\"a}ffer~\cite{PS89} and Peleg and Ullman~\cite{PU89}, has been studied extensively, with applications ranging from routing in networks (e.g. \cite{AP95,TZ05}) to solving linear systems (e.g. \cite{ST04a,EEST08}).  Many of these applications, especially in distributed computing, arise by modeling computer networks or distributed systems as graphs.  But one aspect of distributed systems that is not captured by the above spanner definition is the possibility of \emph{failure}.  We would like our spanner to be robust to failures, so that even if some nodes fail we still have a spanner of what remains.  More formally, $G'$ is an $r$-fault tolerant $k$-spanner of $G$ if for every set $F \subseteq V$ with $|F| \leq r$, the spanner condition holds for $G \setminus F$, i.e. for all $u,v \in V \setminus F$ we have $d_{G' \setminus F}(u,v) \leq k \cdot d_{G \setminus F}(u,v)$.

This notion of fault-tolerant spanners was first introduced by Levcopoulos, Narasimhan, and Smid~\cite{LNS98} in the context of geometric spanners (the special case when the vertices are in Euclidean space and the distance between two points is the Euclidean distance).  They provided both size and weight bounds for $(1+\epsilon)$-spanners, which were later improved by Lukovski~\cite{Luk99} and Czumaj and Zhao~\cite{CZ03}.
The first result on fault-tolerant spanners for general graphs,
by Chechik, Langberg, Peleg, and Roditty~\cite{CLPR09},
constructs $r$-fault tolerant $(2k-1)$-spanners with size $O(r^2 k^{r+1} \cdot n^{1+ 1/k} \log^{1- 1/k} n)$, for any integer $k \geq 1$.  Since it has long been known how to construct $(2k-1)$-spanners with size $O(n^{1+ 1/k})$ (see e.g~\cite{ADDJS93}), this means that the extra cost of $r$-fault tolerance is $O(r^2 k^{r+1})$.  While this is independent of $n$, it grows rapidly as the number of faults $r$ gets large.
We address an important question they left open
of improving this dependence on $r$ from exponential to polynomial.

Nontrivial absolute bounds on the size of a $k$-spanner are possible only when the stretch $k \geq 3$.  For $k=2$, there are graphs with $\Omega(n^2)$ edges for which every edge must be included in the spanner (e.g., a complete bipartite graph).  So the common approach is to provide relative bounds, namely, design approximation algorithms for the problem of computing a minimum size/weight $r$-fault tolerant $2$-spanner.  In this context one assumes that all edges have unit length, so the size equals the weight.  Without fault tolerance, the problem is reasonably well understood: there are algorithms that provide an $O(\log n)$-approximation~\cite{KP94, EP01} (or, with some extra effort, an $O(\log (|E|/|V|))$-approximation), and the problem is NP-hard to approximate better than $\Omega(\log n)$~\cite{Kortsarz01}.
For the $r$-fault tolerant $2$-spanner problem,
Dinitz and Krauthgamer~\cite{DK10} recently gave an $O(r \log n)$-approximation.
However, they did not provide evidence that this loss of $r$ was necessary,
an issue that we address in this paper.

\subsection{Results and Techniques}

\paragraph{Stretch bounds $k \geq 3$.}
Here, our main result is a new $r$-fault tolerant $k$-spanner with size
that depends only polynomially on $r$,
thereby improving over the exponential dependence by Chechik et al.~\cite{CLPR09}.

\begin{theorem} \label{thm:FTkSpanner}
For every graph $G=(V,E)$ with positive edge-lengths and odd $k \geq 3$, there is an $r$-fault tolerant $k$-spanner with size $O(r^{2-\frac{2}{k+1}} n^{1+\frac{2}{k+1}} \log n)$.
\end{theorem}

In fact, we prove something slightly stronger: a general conversion theorem that turns any algorithm for constructing $k$-spanners with size $f(n)$ into an algorithm for constructing $r$-fault tolerant $k$-spanners with size $O(r^3 \log n \cdot f(2n/r)$.  Applying this conversion to the well-known greedy spanner algorithm
(see e.g.~\cite{ADDJS93}) immediately yields Theorem \ref{thm:FTkSpanner}.

%Since we prove a general conversion theorem, our techniques are quite different than those used by
At a high level, Chechik et al.~\cite{CLPR09} apply the spanner construction
of Thorup and Zwick~\cite{TZ05} to every possible fault set,
eventually taking the union of all of these spanners.
They show, through a rather involved analysis that relies on specific
properties of the Thorup-Zwick construction, that taking a union over as many
as $O(n^r)$ spanners increases the size bound only by an $O(r^2 k^r)$ factor.  Our conversion technique, on the other hand, is extremely general.  Inspired by the \emph{color-coding} technique of Alon, Yuster, and Zwick~\cite{AYZ95} and its recent incarnation in designing data structures and oracles~\cite{WY10}, we randomly sample nodes to act as a fault set, and then apply a generic spanner algorithm on what remains.  Our sampling dramatically oversamples nodes --- instead of fault sets of size $r$, we end up with sampled fault sets of size approximately $(1-\frac1r)n$.  This allows us to satisfy many fault sets of size $r$ with a single iteration of the generic algorithm.  The size bound follows almost immediately.

\paragraph{Stretch $k=2$ (and assuming unit-length edges).}
Here, our main result is an approximation algorithm with ratio that is \emph{independent} of $r$.  Our algorithm actually works in an even more general setting, where the graph is directed and edges have \emph{costs} $c_e:E\to{\mathbb R}_{\geq 0}$.  The goal is to find an $r$-fault tolerant $2$-spanner of minimum total cost.  We refer to this problem as {\sc Minimum Cost $r$-Fault Tolerant $2$-Spanner}.

\begin{theorem}\label{thm:FT2spanner}
For every $r \leq n$, there is a (randomized) $O(\log n)$-approximation algorithm for {\sc Minimum Cost $r$-Fault Tolerant $2$-Spanner}.
\end{theorem}

This improves over the previously known $O(r \log n)$-approximation~\cite{DK10}.
Similarly to~\cite{DK10}, we design a flow-based linear programming (LP) relaxation
of the problem and then apply a rounding scheme
that uses randomization at the vertices, rather than naively at the edges.
However, the relaxation used by~\cite{DK10} is not strong enough to achieve
approximation factor independent of $r$;
even simple graphs (such as the complete graph with unit costs)
have integrality gaps of $\Omega(r)$.
We thus design a different relaxation,
%which takes care of cases like the complete graph.  There are other instances, though, on which even the new relaxation has an $\Omega(r)$ integrality gap.  To overcome this, we add to the relaxation
and add to it a large family of constraints that are essentially the \emph{knapsack-cover inequalities} of Carr, Fleischer, Leung, and Phillips~\cite{CFLP00},
adapted to our context.
With these additional constraints, we are able to show that the simple rounding scheme devised in \cite{DK10} now achieves an $O(\log n)$-approximation.

We further show that the integrality gap is at most $O(\log \Delta)$, where $\Delta$ is the maximum degree of the graph, in the special case where all edge costs are $1$.  Note that this bound is at least as good as the $O(\log n)$ bound (and possibly better).  We prove this by a more careful analysis of essentially the same randomized rounding scheme using the \Lovasz Local Lemma.  This makes the result non-algorithmic -- it only shows that the rounding scheme succeeds with a positive probability.

\paragraph{Distributed versions of our algorithms.}
Finally, one feature that is shared by both the $k=2$ and the $k \geq 3$ case is that the algorithms are \emph{local} (assuming that the generic algorithm used by the conversion theorem is itself local).  To show this formally, we provide distributed versions of the algorithm in the $\mathcal{LOCAL}$ model of distributed computation.  The $\mathcal{LOCAL}$ model is a standard message-passing model in which in each round, every node is allowed to send an unbounded-size message to each of its neighbors~\cite{Peleg_book}.  While the unbounded message-size assumption may not be realistic, this model captures locality in the sense that in $t$ rounds, each node has knowledge of, and is influenced by, only the nodes that are within (hop-)distance $t$ of it.

Assuming that the underlying generic spanner algorithm is distributed in this sense, our general conversion theorem trivially provides a distributed algorithm since the failure sampling is done independently by every edge.  Designing a distributed version of the $r$-fault tolerant $2$-spanner algorithm is not quite as simple, since our centralized algorithm uses the Ellipsoid method to solve a linear program that has an exponential number of constraints.  While there is a significant amount of literature on solving linear programs in a distributed manner, much of the time strong assumptions are made about the structure of the linear program.  In particular, it is common to assume that the LP is a \emph{positive} (i.e.~a packing/covering) LP.  Unfortunately the LP relaxation that we use is not positive, even for $r=0$, so we cannot simply use an off-the-shelf distributed LP solver.  Instead, we leverage the fact that the LP itself is ``mostly'' local --- we partition the graph into clusters, solve the LP separately on each cluster, and then repeat this process several times, eventually taking the average values.  This technique is quite similar to the work of Kuhn, Moscibroda, and Wattenhofer~\cite{KMW06}, who showed how to approximately solve positive LPs using the graph decompositions of Linial and Saks~\cite{LS93}.  We construct padded decompositions using a variant of the methods developd by Bartal~\cite{Bartal96} and by Linial and Saks~\cite{LS93}. Combining this distributed methodology for solving the LP relaxation together with the obvious distributed implementation of the aforementioned rounding scheme, we obtain the following distributed $O(\log n)$-approximation.

\begin{theorem}
There is a randomized algorithm that takes $O(\log^2 n)$ rounds and gives an $O(\log n)$-approximation for {\sc Minimum Cost $r$-Fault Tolerant $2$-Spanner} in the $\mathcal{LOCAL}$ model of distributed computation.
\end{theorem}

 \section{General $k$} \label{sec:general}
In this section we give our construction of $r$-vertex-tolerant $k$-spanners (with arbitrary edge-lengths).
For each $F \subseteq V$ with $|F| \leq r$, we let $E_F$ denote the edges of $G \setminus F$, i.e.~$E_F = \{\{u,v\} \in E : u,v\not\in F\}$.  We first give a general conversion theorem that turns any $k$-spanner construction into an $r$-fault tolerant $k$-spanner construction at an extra cost of at most $poly(r) \cdot \log n$.  This conversion actually works fine even when the underlying spanner construction is randomized, but since good deterministic constructions exist we will assume for simplicity that the underlying construction is deterministic.  We say that an event happens with high probability if it happens with probability at least $1-\frac{1}{n^C}$ for constant $C$ that can be made arbitrarily large (at the cost of increasing the constants hidden by $O(\cdot)$ notation).

\begin{theorem} \label{thm:conversion}
If there is an algorithm that on every graph builds a $k$-spanner of size $f(n)$, then there is an algorithm that on any graph builds with high probability an $r$-fault tolerant $k$-spanner of size $O(r^3 \log n \cdot f(\frac{2n}{r}))$.
\end{theorem}
\begin{proof}
Our algorithm is simple: in each iteration, we independently add each vertex to a set $J$ with probability $p = 1-1/r$, and then use the given algorithm to build a $k$-spanner on the remaining graph $G \setminus J$.  If $r=1$ then we can set $p = 1/2$, which will just affect the constants in the $O(\cdot)$.  We do this for $\alpha = \Theta(r^3 \log n)$ iterations, each independent of the others.  Let $H$ be the graph obtained by taking the union of the iterations.

We first bound the size of $H$.  Without loss of generality we can assume that $r \leq n^{2/3}$, since when $r > n^{2/3}$ the claimed size bound is larger than $n^2$ and thus trivially true.  In each iteration, the expected number of vertices in $G \setminus J$ is $n/r$. By a simple Chernoff bound, the probability that a given iteration has more than $2n / r$ vertices in $G \setminus J$ is at most $e^{-(1/3) n/r} \leq e^{-(1/3) n^{1/3}}$.  Since there are only $\alpha = O(r^3 \log n) \leq O(e^{3\ln n} \log n)$ iterations, we can take a union bound over the iterations and get that with high probability the number of vertices in $G \setminus J$ is at most $2n / r$ in \emph{every} iteration.  Thus the total size of $H$ is at most $O(\alpha \cdot f(\frac{2n}{r}))$.  Now we just need to prove that this algorithm results in a valid $r$-fault tolerant $k$-spanner for $\alpha = O(r^3 \log n)$.

For each $F \subseteq V$ with $|F| \leq r$, let $E'_F$ be the edges in $E_F$ for which the shortest path in $G \setminus F$ between the endpoints is just the edge.  More formally, $E_F = \{\{u,v\} \in E_F : d_{G \setminus F}(u,v) = \ell(\{u,v\})\}$.  It is easy to see that it is sufficient for there to be a path of length at most $k \cdot d_{G\setminus F}(u,v)$ between $u$ and $v$ in $G \setminus F$ for every $F \subseteq V$ with $|F| \leq r$ and $\{u,v\} \in E'_F$.  This is because for a given failure set $F$, if we distort the distances of all remaining edges that are actually part of shortest paths by at most $k$, then we distort the distances of all pairs by at most $k$ (since each edge on the shortest path is distorted by at most $k$).  So we consider a particular such $F$ and $\{u,v\}$ and upper bound the probability that there is no stretch-$k$ path between $u$ and $v$ in $G \setminus F$.

Suppose that in some iteration neither $u$ nor $v$ is in $J$, but all of $F$ is in $J$.  Then since $\{u,v\} \in E'_F$, the spanner that we build on $G \setminus J$ contains a path between $u$ and $v$ of length at most $k \cdot d_{G \setminus J}(u,v) = k \cdot \ell(\{u,e\}) = k \cdot d_{G \setminus F}(u,v)$.  Obviously this path also exists in $G \setminus F$, since $F \subseteq J$.  So if this happens then $H$ is valid for $\{u,v\}$ and $F$.  The probability that this happens in a particular iteration is clearly $(1-p)^2 \cdot p^r$, which is at least $1/(4r^2)$ as long as $r \geq 2$ (if $r=1$ then this probability it $1/8$, which does not significantly affect the results).  Thus the probability that this never happens in any iteration is at most $(1-\frac{1}{4r^2})^{\alpha} \leq e^{-\alpha / 4r^2}$, so if we set $\alpha = \Theta(r^3 \log n)$ this becomes less than $1/n^{C(r+2)}$ for arbitrarily large constant $C$.  Now taking a union bound over all $\{u,v\}$ and $F$ gives the theorem.
\end{proof}

\begin{corollary} \label{cor:absolute}
For every graph $G=(V,E)$ with nonnegative edge lengths $\ell : E \rightarrow \mathbb{R}_{\geq 0}$
and every odd $k\ge 1$, there is a polynomial time algorithm that with high probability constructs a $r$-vertex-tolerant $k$-spanner with at most $O(r^{2-\frac{2}{k+1}} n^{1+\frac{2}{k+1}}\log n)$ edges.
\end{corollary}
\begin{proof}
Alth\"ofer et al.~\cite{ADDJS93} showed that the simple greedy spanner construction has size at most $O(n^{1+\frac{2}{k+1}})$.  Applying Theorem~\ref{thm:conversion} to this construction completes the proof.
\end{proof}

Since Theorem~\ref{thm:conversion} applies to any $k$-spanner construction, we can apply it to \emph{distributed} spanner constructions.  We assume that every node knows $r$, the desired amount of fault tolerance.

\begin{theorem} \label{thm:abs_dist}
If there is a distributed algorithm $A$ that on every graph builds a $k$-spanner of size $f(n)$ in $t(n)$ rounds, then there is a distributed algorithm that on any graph builds with high probability an $r$-fault tolerant $k$-spanner of size $O(r^3 \log n \cdot f(2n/r))$ in $O(r^3 \log n \cdot t(n))$ rounds.
\end{theorem}
\begin{proof}
The algorithm is simple: $O(r^3 \log n)$ times, each edge independently decides whether or not to join $J$ with probability $1-1/r$, and then $A$ is run on the remainder.  This obviously takes at most $O(r^3 \log n \cdot t(n))$ rounds, and the analysis of Theorem~\ref{thm:conversion} proves the desired size bound.
\end{proof}

\begin{corollary} \label{cor:abs_dist}
There is a distributed algorithm in the $\mathcal{LOCAL}$ model that in $O(k r^3 \log n)$ rounds constructs with high probability an $r$-fault tolerant $k$-spanner with at most $O(k r^{2 - \frac{2}{k+1}} n^{1+\frac{2}{k+1}} \log n)$ edges.
\end{corollary}
\begin{proof}
Apply Theorem~\ref{thm:abs_dist} to the distributed deterministic spanner construction of Derbel, Gavoille, Peleg, and Viennot~\cite{DGPV08}, which has size $O(k n^{1+\frac{2}{k+1}})$ and runs in $O(k)$ rounds.
\end{proof}

\section{Unit-Length $r$-Fault Tolerant $2$-Spanner}

We now move from general $k$ to the specific case of $k=2$.  It is easy to see (and has long been known) that no non-trivial absolute bounds on the size a $2$-spanner are possible, so following previous work, we instead consider the approximation version.  In this section we will mostly work in the directed setting in which every edge $e$ has an arbitrary cost $c_e\ge 0$.
This is obviously more general than the undirected, unit-cost setting considered in Section~\ref{sec:general}; we can work in this setting because of our additional assumptions that $k=2$ and edge lengths are unit.  Recent work of Dinitz and Krauthgamer~\cite{DK10} achieves approximation ratio $O(r \log n)$ for the unit-length $r$-fault tolerant $2$-spanner problem, and an $O(r \log \Delta)$ upper bound on the integrality gap (where $\Delta$ is the maximum degree).
Here we improve these results to $O(\log n)$ and $O(\log \Delta)$ (for all $r$)
via a different LP relaxation, and also provide a distributed implementation.

\subsection{The Previous LP Relaxation}

The relaxation in~\cite{DK10} uses, at a high level, a characterization of $r$-fault tolerant $2$-spanners
based on flows where ``for every set of $r$ faults, it is possible to send one unit of (integral) flow from $u$ to $v$ along paths of length at most $2$ for any edge $(u,v)$ still present in the graph once the faults have been removed''.  More formally, for each $(u,v) \in E$ let $\mathcal{P}_{u,v}$ denote the paths of length \emph{exactly} two from $u$ to $v$, so $\mathcal{P}_{u,v} \cup \aset{(u,v)}$ is the set of all paths of length at most $2$.  As in Section~\ref{sec:general}, for any possible fault set $F \subseteq V$ with $|F| \leq r$ let $E_F$ be the set of edges in $E$ with neither endpoint in $F$.  Let $\mathcal{P}^F_{u,v}$ be the subset of $\mathcal{P}_{u,v} \cup \aset{(u,v)}$ that still survives in $E_F$.
The integer program (IP) used by Dinitz and Krauthgamer~\cite{DK10} is presented
below as IP~\eqref{IP:old}.

\begin{equation} \label{IP:old}
\framebox{ $
\begin{array}{lll}
  \min & \displaystyle \sum_{e \in E} c_e x_e
  \\
  \mathrm{s.t.}
  & \displaystyle \sum_{P \in \mathcal P_{u,v}^F:\ e \in P} f_P^F \leq
  x_e & \forall F \subseteq V : |F| \leq r,\
  \forall (u,v) \in E_F,\
  \forall e \in E_F
  \\
  & \displaystyle \sum_{P \in \mathcal P_{u,v}^F} f_P^F \geq 1
  & \forall F \subseteq V : |F| \leq r,\ \forall (u,v)\in E_F
  \\
  & \displaystyle x_e \in \{0,1\} & \forall  e \in E
  \\
  & \displaystyle f_P^F \in \{0,1\} & \forall F \subseteq V : |F| \leq r, \ \forall (u,v) \in E_F, \ \forall P \in \mathcal P_{u,v}^F
\end{array}
$ }
\end{equation}

This formulation has capacity variables $x_e$ for every edge $e$, flow variables $f_P^F$ for every possible fault set $F$ and every path $P \in \mathcal{P}^F_{u,v} \cup (u,v)$ (for every $(u,v) \in E$), and constraints that require flows to obey the capacities and still send one unit of flow for every possible fault set.   Even though there are an exponential number of both constraints and variables, it is solvable in polynomial time~\cite{DK10}.

While IP~\eqref{IP:old} is the obvious integer programming formulation of the $r$-fault tolerant $k$-spanner problem, its straightforward relaxation to a linear program is not strong enough to give an approximation that is independent of $r$ (despite having an exponential number of both constraints and variables).  An easy way to see this is by considering the complete graph.  On the complete graph, every vertex obviously needs at least $r$ incoming and outgoing edges, or else it could be isolated with less than $r$ faults.  So on $K_n$ the optimum spanner has size at least $rn$.  On the other hand, when we relax the integrality constraints we can set the capacity of every edge to $1/ (n-r-2)$ and still have enough capacity to send one unit of flow from any vertex to any other even after $r$ of them have failed.  So the linear program has cost of only $n^2 / (n-r-2)$, which is $O(n)$ as long as $r < cn$ for some constant $c < 1$.  Thus the integrality gap of the relaxation is $\Omega(r)$ for an extremely wide range of $r$.

\subsection{A New LP Relaxation}

To get around this problem, we will use a different relaxation based on \emph{weighted} flow.  Before we give our formulation, we first prove a simple and useful characterization of $r$-fault-tolerant $2$-spanners:

\begin{lemma} \label{lem:FT_char}
For any (directed) graph $G = (V,E)$, a subgraph $H = (V, E')$ is an $r$-fault tolerant $2$-spanner if and only if for every $(u,v)$ in E either $(u,v) \in E'$ or there are at least $r+1$ paths of length $2$ from $u$ to $v$ in $E'$
\end{lemma}
\begin{proof}
Let $H$ be an $r$-fault tolerant $2$-spanner of $G$, and for the sake of contradiction assume that there is some $(u,v) \in E$ that is not in $E'$ and for which there are at most $r$ paths of length $2$ from $u$ to $v$.  Let $W \subseteq V$ be the vertices that are the midpoints of these paths.  Then if we let our fault set $F$ be $W$, in the remaining graph $H \setminus W$ there is no $u-v$ path, while in $G \setminus W$ the edge $(u,v)$ still exists.  Thus $H$ is not an $r$-fault tolerant $2$-spanner, giving the contradiction.

For the other direction, suppose that for every $(u,v) \in E$ either $(u,v) \in E'$ or there are at least $r+1$ paths of lengths $2$ from $u$ to $v$.  Let $F \subseteq V$ with $|F| \leq r$ be some fault set.  We need to show that $H$ is a valid $2$-spanner for $G \setminus F$, so let $(u,v) \in E$ with $u,v \not\in F$ be an arbitrary edge in $G \setminus F$.  If $(u,v) \in E'$ then obviously $H$ preserves its distance exactly, and  if $(u,v) \not\in E'$ then by assumption there are at least $r+1$ paths from $u$ to $v$ of length $2$ in $E'$.  At most $r$ of the intermediate vertices on those paths can be in $F$, so in $G \setminus F$ there is at least one such path remaining.
\end{proof}

With this lemma in hand, it is easy to see that the following integer program is an exact formulation of the $r$-fault tolerant $2$-spanner problem.

\begin{equation} \label{IP:spanner_simple}
\framebox{ $
\begin{array}{lll}
  \min & \displaystyle \sum_{e\in E} c_e x_e
  \\
  \mathrm{s.t.}
  & \displaystyle \sum_{P\in \mathcal P_{u,v}: e\in P} f_P \le x_e
  & \forall (u,v) \in E,\ \forall e\in E
  \\
  & \displaystyle (r+1) x_{(u,v)} + \sum_{P\in \mathcal P_{u,v}} f_P \ge r+1
  & \forall (u,v) \in E
  \\
  & \displaystyle x_e \in \{0,1\} & \forall e \in E
  \\
  & \displaystyle f_P \geq 0 & \forall (u,v) \in E,\ \forall P \in \mathcal P_{u,v}
\end{array}
$ }
\end{equation}

So now we have a different IP formulation than the one that was used in~\cite{DK10} to get an $O(r \log n)$-approximation.  Unfortunately, it is still not strong enough to yield an approximation ratio independent of $r$; there are still simple examples that give an integrality gap of $\Omega(r)$.  For example, consider a graph with nodes $u$ and $v$ and an edge of cost $M$ from $u$ to $v$ (for some arbitrarily large $M$), together with $r$ nodes $w_1, \dots, w_r$ and an edge of cost $1$ from $u$ to $w_i$ and from $w_i$ to $v$ for all $i \in [r]$.  The set of all $w_i$ nodes is a valid fault set, so the optimum spanner needs to include the $(u,v)$ edge in order to still be valid.  So the optimum spanner has cost at least $M$.  On the other hand, the LP can set $x_e$ to $1$ for all edges $e$ incident on some $w_i$, and set $x_{(u,v)} = 1/(r+1)$.  This has cost of only $M/(r+1) + 2r$.  By setting $M$ large enough, we get a gap of $\Omega(r)$.

We will strengthen the relaxation by adding a set of valid inequalities that are essentially the \emph{knapsack-cover inequalities} of Carr et al.~\cite{CFLP00} applied to this IP.  Let $(u,v) \in E$, and consider some arbitrary subset $W \subseteq \mathcal P_{u,v}$ with $|W| \leq r$.  If $x_{(u,v)} = 0$, then the covering inequality for $(u,v)$ implies that $\sum_{P \in \mathcal P_{u,v}} f_P \geq r+1$, and thus $\sum_{P \in \mathcal P_{u,v} \setminus W} f_P \geq r+1-|W|$.  On the other hand, if $x_{(u,v)} = 1$ then clearly $(r+1 - |W|) x_{(u,v)} \geq r+1-|W|$.  So for all $(u,v) \in E$ and all $W \subseteq \mathcal P_{u,v}$ with $|W| \leq r$, we can add the constraint that $(r+1-|W|) x_{(u,v)} + \sum_{P \in \mathcal P_{u,v} \setminus W} f_P \geq r+1-|W|$.  These are the knapsack-cover inequalities, and when we add them to our IP formulation and relax the integrality constraints we get the following LP relaxation:

\begin{equation} \label{LP:spanner}
\framebox{ $
\begin{array}{lll}
  \min & \displaystyle \sum_{e\in E} c_e x_e
  \\
  \mathrm{s.t.}
  & \displaystyle \sum_{P\in \mathcal P_{u,v}: e\in P} f_P \le x_e
  & \forall (u,v) \in E,\ \forall e\in E
  \\
  & \displaystyle (r+1 - |W|) x_{(u,v)} + \sum_{P\in \mathcal P_{u,v} \setminus W} f_P \ge r+1 - |W|
  & \forall (u,v) \in E, \ \forall W \subseteq \mathcal P_{u,v} : |W| \leq r
  \\
  & \displaystyle 0 \leq x_e \leq 1 & \forall e \in E
  \\
  & \displaystyle f_P \geq 0 & \forall (u,v) \in E,\ \forall P \in \mathcal P_{u,v}
\end{array}
$ }
\end{equation}

We refer to the first type of constraints as \emph{capacity} constraints, the second type as \emph{knapsack-cover} constraints (or inequalities), and the third as \emph{multiplicity} constraints.  This relaxation has a polynomial number of variables but a possibly exponential number of constraints, so we first need to show that we can solve it.  To do this we construct a separation oracle, which allows us to solve it in polynomial time by using the Ellipsoid algorithm.

\begin{lemma} \label{lem:solve}
There is a polynomial time algorithm that solves LP~\eqref{LP:spanner}.
\end{lemma}
\begin{proof}
We want to construct a separation oracle.  Note that there are only a polynomial number of capacity constraints and multiplicity constraints, so we can check them all in polynomial time.  To find a violated knapsack-cover inequality, note that if there is some $(u,v) \in E$ and some $W \subseteq \mathcal P_{u,v}$ that violates the inequality, then the set $W'$ which consists of the $|W|$ paths in $\mathcal P_{u,v}$ with the largest $f_P$ value also violates the inequality.  So for every $(u,v) \in E$, for every $k \in [0,r]$, it suffices to check the constraint for $(u,v)$ and the $k$ paths in $\mathcal P_{u,v}$ with largest flow.  Since $r \leq n$, this takes only polynomial time.
\end{proof}

\subsection{$O(\log n)$-approximation} \label{sec:approx}

We now give the main result of this section.
\begin{theorem} \label{thm:approx}
There is a randomized $O(\log n)$-approximation for {\sc Minimum Cost $r$-Fault Tolerant $2$-Spanner} (for all $r$).
\end{theorem}
\begin{proof}
The first step of the algorithm is to solve LP~\eqref{LP:spanner} using Lemma~\ref{lem:solve}.
We then round the solution using Algorithm \ref{alg:2spanner} below.
(This rounding algorithm was designed in \cite{DK10}, but for a different
relaxation, hence they were forced to set $\alpha = \Theta(r \log n)$ and the analysis therein is not applicable here.)

\begin{algorithm}[H]
\caption{Rounding algorithm for $r$-fault tolerant $2$-spanner.}
\label{alg:2spanner}
Set $\alpha=C\ln n$ (for a large enough constant $C$).
\\
For every $v\in V$ choose independently a random threshold $T_v\in[0,1]$.
\\
Output $E'=\{(u,v)\in E:\ \minn{T_u,T_v} \leq \alpha\cdot x_{u,v}\}$.
\end{algorithm}

We first show that the cost of the solution is likely to be at most $6\alpha$
times the LP value.
The probability than some edge $e$ is selected to be in $E'$ is at most
$2\alpha x_e$, so the expected cost of the solution $E'$ is
$\sum_{e \in E} c_e \cdot 2\alpha x_e = 2\alpha \sum_e c_e x_e$.
By Markov's inequality, the cost of the solution $E'$ exceeds
$6\alpha \sum_e c_e x_e$ with probability at most $1/3$.

We now argue that this algorithm returns a valid $r$-fault tolerant $2$-spanner with high probability.  We say that $E'$ \emph{satisfies} an edge $(u,v)$ if either $(u,v) \in E'$ or $E'$ contains at least $r+1$ length $2$ paths from $u$ to $v$.  By Lemma~\ref{lem:FT_char}, if $E'$ satisfies all edges then it is a valid $r$-fault tolerant $2$-spanner.  Consider some edge $(u,v)\in E$.
Order the paths in $\mathcal P_{u,v}$ in nonincreasing order by their flow values in the LP solution, so $P_i$ is the path with the $i$th largest flow.  Let $W_i = \{P_1, P_2, \dots, P_i\}$, and let $i^* = \max\{i : f_{P_i} \geq 1/\alpha\}$.  If $i^* > r$ then $r+1$ paths have flow value at least $1/\alpha$, so both of the edges in each path have $x$ value at least $1/\alpha$, so they are included in $E'$ with probability $1$.  Thus $(u,v)$ is satisfied with probability $1$.

On the other hand, suppose that $i^* \leq r$.
Let us denote $r' = r+1 - i^* \ge 1$.
By the knapsack-cover constraint for $(u,v)$ and $W_{i^*}$, we know that
\begin{equation*}
r' x_{(u,v)} + \sum_{P \in \mathcal P_{u,v} \setminus W_{i^*}} f_P \geq r'
\end{equation*}

If $r' x_{(u,v)} \geq r' / 2$ then $x_{(u,v)} \geq 1/2$ and thus $(u,v)$ is included in $E'$ with probability $1$, satisfying $(u,v)$.  Otherwise it must be the case that $\sum_{P \in \mathcal P_{u,v} \setminus W_{i^*}} f_P \geq r'/2$.
For $P \in \mathcal P_{u,v}$, let $I_P$ be an indicator for the event that the
$T$ value of the middle vertex is at most $\alpha$ times the flow value $f_P$
(formally, if $P=(u,z,v)$ then $I_P=1_{T_z \leq \alpha f_P}$),
and observe that this event implies that both edges of $P$ are included in $E'$
(because then we have $T_z\leq \minn{x_{(u,z)},x_{(z,v)}}$).
Note that for $P \in W_{i^*}$, we have $I_P=1$ with probability $1$.
For $P \in \mathcal P_{u,v} \setminus W_{i^*}$,
we have $I_P=1$ with probability at least $\alpha f_P\in[0,1]$.
The number of paths from $\mathcal P_{u,v} \setminus W_{i^*}$ included in $E'$
is clearly at least $\sum_{P\in \mathcal P_{u,v} \setminus W_{i^*}} I_P$,
and we can bound that last quantity (which is a sum of independent indicators)
by a Chernoff bound (see e.g. \cite{MR95,DP09}).  Its expectation is
\[
  \EX \Big[\sum_{P\in \mathcal P_{u,v} \setminus W_{i^*}} I_P\Big]
  \geq \sum_{P \in \mathcal P_{u,v} \setminus W_{i^*}} \alpha f_P
  \geq \alpha r' / 2,
\]
so by our choice of $\alpha=C \log n$ for a large enough $C$,
\begin{equation}\label{eq:applyChernoff}
  \Pr\Big[\sum_{P\in \mathcal P_{u,v} \setminus W_{i^*}} I_P \leq \alpha r'/4\Big]
  \leq e^{-\Omega(\alpha r')}
  \leq 1/n^{\Omega(C)}
  \leq 1/n^3.
\end{equation}
Thus, with high probability the total number of length $2$ paths between $u$ and $v$ included in $E'$ is at least $i^* + \alpha r' / 4 \geq r+1$, and thus $(u,v)$ is satisfied.
The theorem follows by taking a union bound over these events
for all edges $(u,v)$,
and the aforementioned event that the solution's cost exceeds
$6\alpha$ times the LP value.
\end{proof}

\subsection{Bounded-Degree Graphs}

When the maximum degree of the graph is bounded by $\Delta$
and the edge costs $c_e$ are all $1$,
we can improve Theorem~\ref{thm:approx} slightly and give an $O(\log \Delta)$-approximation.  We simply change the inflation parameter $\alpha$ in Algorithm~$\ref{alg:2spanner}$ to be $O(\log \Delta)$ instead of $O(\log n)$.  We then need a more careful analysis, using an algorithmic version of the \Lovasz Local Lemma.  

\begin{theorem} \label{thm:bounded_degree}
There is a (randomized) $O(\log \Delta)$-approximation for the (directed) $r$-fault tolerant $2$-spanner problem on graphs in which $c_e = 1$ for all $e \in E$ and the maximum (in and out) degree is at most $\Delta \geq 2$.
\end{theorem}

We shall use the following constructive version of the symmetric \Lovasz Local Lemma, which is an immediate corollary of the nonsymmetric version proved by Moser and Tardos~\cite{MT10}.

\begin{lemma}[Moser and Tardos~\cite{MT10}] \label{lem:LLL}
Let $\mathcal P$ be a finite set of mutually independent random variables in a probability space.  Let $\mathcal A$ be a finite set of events determined by the variables in $\mathcal P$.  Suppose that each $A \in \mathcal A$ is mutually independent of all but at most $d$ other events in $\mathcal A$, and suppose that $\Pr[A] \leq p$ for all $A \in \mathcal A$.  If $ep(d+1) \leq 1$ then there exists an assignment of values to the variables $\mathcal P$ such that no event $A \in \mathcal A$ occurs.  Moreover, there is a randomized algorithm that finds such an assignment in expected time $O(|\mathcal P| + |\mathcal A| \cdot |\mathcal P| / d)$.
\end{lemma}

\begin{proof}[Proof of Theorem~\ref{thm:bounded_degree}]
Consider a directed graph $G$
with unit edges costs $c_e=1$ and vertex degrees bounded by $\Delta$.
Consider a solution to the LP relaxation~\eqref{LP:spanner},
and apply Algorithm~\ref{alg:2spanner} to it but with inflation factor $\alpha = C\log \Delta$ instead of $C \log n$.

For an edge $(u,v) \in E$, let $A_{u,v}$ be the event that $E'$ does not
satisfy this edge, i.e. $(u,v) \not\in E'$ and the graph $G'=(V,E')$
has less than $r+1$ paths of length $2$ from $u$ to $v$.
%In other words, $A_{u,v}$ is the event that $H$ is not
%a valid $r$-fault tolerant $2$-spanner for the edge $(u,v)$.
The analysis of Theorem~\ref{thm:approx} shows
(after modifying~\eqref{eq:applyChernoff} with our new value of $\alpha$), that
\[
  \Pr[A_{u,v}] \leq e^{-\Omega(\alpha)} \leq 1/\Delta^{\Omega(C)}.
\]
Furthermore, note that $A_{u,v}$ depends only on the random variables $T_z$
for $z \in (N^+(u) \cap N^-(v)) \cup \{u\}$.
Here and throughout, $N^+(u)$ and $N^-(u)$ denote the out-neighbors
and in-neighbors of $u\in V$, respectively.
Observe that $A_{u,v}$ is independent of all but $\Delta^3$ other events
$A_{u', v'}$, simply because
there are at most $\Delta$ choices for each of $z$, $u'$, and $v'$.

We could now apply Lemma~\ref{lem:LLL} to these events.  The underlying mutually independent random variables $\mathcal P$ would be the $T_u$ variables, and the ``bad events" $\mathcal A$ would be the events $A_{u,v}$.  This would give us an algorithm that in polynomial time returned a valid $r$-fault tolerant $k$-spanner, but we also need a bound on the cost of this spanner.  The analysis via Markov's inequality in Theorem~\ref{thm:approx} is too weak now,
because when we apply the algorithm of Lemma~\ref{lem:LLL} we change the overall distribution in a way that might destroy the cost bound.  We need to integrate the cost analysis into the events that Lemma~\ref{lem:LLL} is applied to, so at a high level we employ a more local approach where the cost of $E'$ is split among the vertices and events bounding the cost are added to the $A_{u,v}$ events.
More specifically, we shall create many events, each of which controls
how the cost of $E'$ compares \emph{locally} with the cost of the LP,
and then apply the Local Lemma to the new events together with the $\{A_{u,v}\}$ events.
A formal argument follows.

For each vertex $u \in V$, let the random variable $Z_u^+$ be the
number of outgoing edges $(u,v)$ for which $T_v \leq \alpha \cdot x_{u,v}$,
and let $Z_u^-$ be the number of incoming edges $(v,u)$ for which
$T_v \leq \alpha \cdot x_{u,v}$.
Informally, $Z_u^+ + Z_u^-$ is the number of edges incident to $u$
whose inclusion in $E'$ can be charged to their other endpoint.
The algorithm's cost is $|E'| \leq \sum_{u \in V} (Z_u^+ + Z_u^-)$,
since every edge $(u,v)$ included in $E'$ adds $1$ to
either $Z_u^+$ or $Z_v^-$ (or both).

For each vertex $v\in V$, let $B_u$ be the event that
$Z_u^+ + Z_u^- > 4 \alpha (\sum_{(u,v) \in E} x_{u,v} + \sum_{(v,u) \in E} x_{v,u})$.
We would like to show that this event happens only with small probability.
Note that
$\EX[Z_u^+]
  = \sum_{(u,v) \in E} \min\{\alpha x_{u,v}, 1\}
  \leq \alpha \sum_{(u,v) \in E} x_{u,v}$,
so by a Chernoff bound (see e.g. \cite{MR95,DP09}) we get
\begin{equation*}
\Pr\Big[Z_u^+ > 2\alpha \sum_{(u,v) \in E} x_{u,v}\Big]
  \leq e^{-(1/3)(C \ln \Delta) \sum_{(u,v) \in E} x_{u,v}} \leq \Delta^{-C/3},
\end{equation*}
where in the final inequality we assume there is at least one outgoing
edge from $u$ and thus $\sum_{(u,v) \in E} x_{u,v} \geq 1$
(since otherwise $Z_u^+=0$ with probability $1$).
Using a similar argument to bound $Z_u^-$, we get
% Similarly, $\EX[Z_u^-] \leq \alpha \sum_{(v,u) \in E} x_{v,u}$
% and thus $\Pr[Z_u^- > 2\alpha \sum_{(v,u) \in E} x_{v,u}] \leq \Delta^{C/3}$.
% In order for $B_u$ to occur, at least one of these two events must occur,
% and thus
\begin{equation*}
  \Pr[B_u]
  \leq \Pr\Big[Z_u^+ > 2\alpha \sum_{(u,v) \in E} x_{u,v}\Big]
    + \Pr\Big[Z_u^- > 2\alpha \sum_{(v,u) \in E} x_{v,u}\Big]
  \leq 2\Delta^{-C/3}.
\end{equation*}

We now apply Lemma~\ref{lem:LLL} to the events and $A_{u,v}$ and $B_u$.
%and $A_u=\cup_{(u,v)\in E} A_{u,v}$
Note that $B_u$ depends only on the random variables $T_z$ for
$z \in N^+(u) \cup N^-(u)$,
and recall that $A_{u,v}$ depends only on $T_z$ for $z\in N^+(u)\cap N^-(v)$.  
Thus each event is mutually independent of all but $O(\Delta^3)$ other events ---
for an event $A_{u,v}$ we exclude at most $\Delta^3$ events $A_{u', v'}$
and at most $2\Delta^2$ events $B_{u'}$;
for an event $B_u$ we exclude at most $4\Delta^2$ events $B_{u'}$ and
at most $2\Delta^3$ events $A_{u',v'}$.
We can thus apply the Lemma~\ref{lem:LLL} with dependency parameter $d = O(\Delta^3)$,
because by setting sufficiently large $C$,
the probability of each event is at most a suitable $p=\Delta^{-\Omega(C)}<1/e(d+1)$.
Since the number of events is at most $O(n^2)$ and the number of underlying variables is only $n$, we conclude that there is a polynomial time algorithm to find the underlying variables $T_u$ so that
none of the events $A_{u,v}$ and $B_u$ occur.  This implies that $G'=(V,E')$ is an $r$-fault tolerant $2$-spanner of $G$ of cost
\begin{equation*}
  |E'|
  \leq \sum_{u \in V} \left(Z_u^+ + Z_u^-\right)
%  \leq \sum_{u \in V} 4\alpha \left(\sum_{(u,v) \in E} x_{u,v}
%    + \sum_{(v,u) \in E} x_{v,u}\right)
  \leq 8\alpha \sum_{(u,v)\in E} c_{u,v} x_{u,v}
  \leq O(\log \Delta) \cdot \mathrm{ LP},
\end{equation*}
which proves Theorem~\ref{thm:bounded_degree}.
\end{proof}

\subsection{Distributed Construction}

We now show how to adapt and use the $O(\log n)$-approximation that we designed in Section~\ref{sec:approx} to give a distributed $O(\log n)$-approximation.  We will assume that communication along an edge is bidirectional, even if the graph is directed.  The main problem that we run into when trying to design a distributed algorithm based on Algorithm~\ref{alg:2spanner} is solving the linear program.  If we had a solution, and every vertex knew the $x_e$ value of its incident edges, then we would be done; the rounding scheme in Algorithm~\ref{alg:2spanner} is entirely local, so every vertex $v\in V$ would just locally pick its threshold $T_v$ and include the appropriate edges.  If we want both endpoints of an edge to know that it has been included in the spanner, we can then just have every vertex tell all of its neighbors (in a single round) which edges it bought based on its threshold.

In order to (approximately) solve the LP we partition the graph into clusters, solve the LP separately on each cluster, and then repeat this process several times, eventually taking the average.  This technique is quite similar to the work of Kuhn, Moscibroda, and Wattenhofer~\cite{KMW06}, who showed how to approximately solve \emph{positive} LPs using the graph decompositions of Linial and Saks~\cite{LS93}.

The fundamental tool that we will use is the ability to quickly compute
a good \emph{padded decomposition},
which is a basic tool in metric embeddings, but has found numerous applications
in approximation and online algorithms (e.g. for network design problems).
This notion is essentially a version of low-diameter decompositions,
such as a sparse covers \cite{AP90}.
This specific version was (probably) introduced by Rao \cite{Rao99},
who observed it is can be derived from an earlier construction
of Klein, Plotkin and Rao \cite{KPR93}.
An explicit formulation of padded decompositions appeared only later,
in \cite{KL03,GKL03}, and used a construction of Bartal \cite{Bartal96}.
The definition given below is actually a special case of the usual notion,
where the so-called padding requirement is a unit radius around each vertex,
i.e. just the vertex's neighborhood.

Let $\mathcal T = \mathcal T(V)$ denote the set of all partitions of $V$
(irrespective of the graph structure).  For a partition $P\in \mathcal T$, we call each set $C\in P$ a \emph{cluster}.  Let $G'$ be the \emph{undirected} graph corresponding to $G$, and define the diameter of $C$ to be $\diam(C) = \max_{u,v \in V} d_{G'}(u,v)$
(this is usually called weak diameter, because it corresponds to the
shortest $u-v$ path in $G'$, possibly going out of $C$ along the way).
Finally, for $x\in V$ and a partition $P\in \mathcal T$,
we let $P(x)$ denote the cluster of $P$ that contains $x$.

\begin{definition} \label{def:padded}
A \emph{padded decomposition} of $G$ is a probability measure $\mu$
on $\mathcal T$ that satisfies the following two conditions:
\begin{enumerate} \compactify
\item For every $P\in \supp(\mu)$ and every $C \in P$
we have $\diam(C) \leq  O(\log n)$.
\item For every $x\in V$ we have
$\Pr_{P \sim \mu}[N(x) \subseteq P(x)] \geq 1/2$.
\end{enumerate}
\end{definition}

It is known that every metric space admits such a padded decomposition,
and there are polynomial-time randomized algorithms to sample
from such a decomposition~\cite{Bartal96,FRT04}.
It is convenient to assign to each cluster a vertex,
called the \emph{cluster center}.
One could always choose an arbitrary vertex in the cluster
(e.g. one whose identifier is the smallest).
The next lemma is a straightforward adaptation of the construction of
Bartal~\cite{Bartal96} to the distributed context;
it can also be viewed as a slight modification to the graph decompositions
of Linial and Saks~\cite{LS93}.

\begin{lemma} \label{lem:distributed_padded}
There is an algorithm in the $\mathcal{LOCAL}$ model that runs in $O(\log n)$ rounds and \whp\ samples from a padded decomposition, so that every vertex knows the cluster containing it, meaning all other vertices in the same cluster.  Every cluster $C$ also has a cluster center $v \in V$ (which is not necessarily in the cluster) with the property that $\diam(C\cup\{v\}) \leq O(\log n)$.
\end{lemma}
\begin{proof}
The construction of Bartal~\cite{Bartal96} is simple,
and is usually described iteratively.
(As mentioned above, the padding property is not formally proved there,
but it can be derived from the analysis therein, see also \cite{KL03,GKL03}).
Working in the metric completion of $G$ (so removing vertices does not change distances), repeat the following procedure until every vertex has been assigned to some cluster:
Pick an arbitrary vertex $u$ from those that have not yet been assigned a cluster.  Randomly pick a radius $r_u$ from the geometric distribution with some constant parameter $p>0$.  Create a new cluster consisting of $u$ and all unclustered vertices that are within distance $r_u$ of $u$.

While this procedure is phrased iteratively, it quite obviously can be made distributed with only minor changes.  First, every vertex $u \in V$ locally chooses a value $r_u$ from the geometric distribution with parameter $p$.  Then every node $u$ simultaneously sends a message containing the ID of $u$ to all nodes within distance $\min\{r_u, O(\log n)\}$ of $u$.  Note that this take only $O(\log n)$ rounds, and with high probability  $\max_u \{r_u\} \leq O(\log n)$ (in fact we could truncate the exponential distribution so that $r_u$ is always less than $O(\log n)$, since the analysis of~\cite{KL03} shows that this does not significantly affect the padding probability).  Now every node chooses as a cluster center the sender with the smallest ID (i.e.~the sender that comes earliest in the lexicographic ordering) of the vertices whose messages it received.  The only difference between the output of this algorithm and Bartal's algorithm is that in Bartal only unclustered nodes can be the center of a new cluster, while in our variation every vertex (in lexicographic order) gets the chance to create a cluster (which it might not be a member of itself).  It is well known (see e.g.~\cite{KL03,GKL03}) that this change does not affect anything in the analysis.

We remark that the construction above has a natural choice of cluster centers.
Under this choice, a cluster $C$ might not contain its center $v\in V$,
but $\diam(C\cup\{v\}) \leq O(\log n)$, which is sufficient for our purposes.
\end{proof}

Now that we can construct padded decompositions,
we want to use them to decompose LP~\eqref{LP:spanner} into ``local'' parts.
Let $P$ be a partition sampled from $\mu$. For each cluster $C \in P$, let $N(C)$ denote the set of vertices in $V \setminus C$ that are adjacent to at least one vertex in $C$, let $\delta(C) \subseteq E$ be all edges with one endpoint in $C$ and one endpoint not in $C$, and let $E(C) \subseteq E$ be the set of edges with both endpoints in $C$.  Let $G(C)$ be the subgraph of $G$ induced by $C \cup N(C)$. We define $\LP(C)$ to be LP~\eqref{LP:spanner} for $G(C)$, but where edges in $\delta(C)$ are modified to have cost $0$.

Let $\LP^*$ be the value of an optimal solution to LP~\eqref{LP:spanner},
and let $\LP^*(C)$ be the value of an optimal solution to $\LP(C)$.

\begin{lemma} \label{lem:decomposed_cost}
$\sum_{C \in P} \LP^*(C) \leq \LP^*$ for every partition $P\in \mathcal T$.
\end{lemma}
\begin{proof}
Let $\tuple{x,f}$ be an optimal fractional solution to LP~\eqref{LP:spanner}.
We want to use this solution to build fractional solutions to $\LP(C)$ for all $C \in P$ whose total cost is at most $\LP^*$.
For each cluster $C \in P$, define a solution $\tuple{x^C,f^C}$ for $\LP(C)$
as follows:
Let $x^C_e = x_e$ if $e \in E(C)$ and let $x^C_e = 1$ if $e \in \delta(C)$.  Note that this already satisfies all of the knapsack-cover constraints for edges in $\delta(C)$.
For edges $(u,v)\in E(C)$, note that every path in $\mathcal P_{u,v}$
appears in $G(C)$, so we can set $f^C_P = f_P$ for these paths.
Since these flows satisfy the knapsack-cover constraints
in LP~\eqref{LP:spanner},
we now satisfy also the knapsack-cover constraints in $\LP(C)$.
All other flows $f^C_P$ (e.g. between vertices in $N(C)$) are set to $0$,
and obviously the capacity constraints are satisfied,
hence $\tuple{x^C, f^C}$ is a feasible solution to $\LP(C)$.

Since in $\LP(C)$ the edges in $\delta(C)$ have cost $0$,
and every edge of $E$ is in $E(C)$ for at most one cluster $C$,
\begin{equation*}
  \sum_{C \in P} \LP^*(C)
  = \sum_{C \in P} \sum_{e \in E(C)} c_e x^C_e
  \leq \sum_{e \in E} c_e x_e = \LP^*,
\end{equation*}
which proves the lemma.
\end{proof}

\smallskip
We can now give our distributed algorithm for
{\sc Minimum Cost $r$-Fault Tolerant $2$-Spanner}:
\smallskip

\begin{algorithm}[H]
\caption{Distributed algorithm for $r$-fault tolerant $2$-spanner.}
\label{alg:distributed_2spanner}
\For{$i\leftarrow 1$ \KwTo $t = O(\log n)$}{
  Sample a partition $P_i$ from $\mu$ using Lemma~\ref{lem:distributed_padded}
  \tcp*{we assume the center of each cluster $C\in P_i$ knows $G(C)$}
%  Each $u \in V$ sends a list of its incident edges to the center of
%  its cluster $P_i(u)$ \\

  The center of each cluster $C\in P_i$ solves $\LP(C)$
  using Lemma~\ref{lem:solve},
  and sends the solution $\tuple{x^{C,i}, f^{C,i}}$ to all vertices in $C$
  \;
}
For each edge $(u,v)\in E$, let $\mathcal I_{(u,v)}=\{i: P_i(u)=P_i(v)\}$
\tcp*{these are the iterations in which both endpoints are in same cluster}

$\tilde x_e \leftarrow \min\{1, \frac{4}{t}\sum_{i \in \mathcal I_e} x^{P_i(e), i}_e\}$
\tcp*{$P_i(e)$ is the cluster of $P_i$ containing both endpoints of $e$}

Round $\tilde x_e$ using Algorithm~\ref{alg:2spanner}
\tcp*{each edge is rounded by its endpoints}
\end{algorithm}

\begin{theorem} \label{thm:distributed_2spanner}
Algorithm~\ref{alg:distributed_2spanner} terminates in $O(\log^2 n)$ rounds
and computes (in expectation) an $O(\log n)$-approximation
to {\sc Minimum Cost $r$-Fault Tolerant $2$-Spanner}.
\end{theorem}
\begin{proof}
We first prove the time bound.  Lemma~\ref{lem:distributed_padded} implies that sampling from $\mu$ takes only $O(\log n)$ rounds, and since the diameter of every cluster is at most $O(\log n)$ the other two steps of the loop also take only $O(\log n)$ rounds.
Since we execute the loop $O(\log n)$ times, the number of rounds needed to
complete the loop is at most $O(\log^2 n)$.
After the loop, each vertex can compute $x_e$ for all incident edges $e$
without any extra communication (since each endpoint of an edge $e$
knows $\mathcal I_e$ and the LP values for that iteration).
Finally, as already pointed out, the rounding of Algorithm~\ref{alg:2spanner} can be done locally, with one extra round used to make sure that both endpoints of an edge know if the edge was included by the rounding.  Thus the total number of rounds is $O(\log^2 n)$, as claimed.

To prove that this algorithm returns an $O(\log n)$-approximation,
we will show that \whp\ the $\tilde x_e$ values it computes form
a feasible solution to LP~\eqref{LP:spanner}
(when appropriate flow values $\tilde f_P$ are chosen)
of cost at most $O(\LP^*)$.
Once we have this, the analysis of Theorem~\ref{thm:approx} implies that
the rounding step outputs (in expectation) a spanner $G'=(V,E')$ whose cost is
$O(\log n) \sum_e c_e \tilde x_e \leq O(\log n) \LP^*$,
which is clearly an $O(\log n)$-approximation as asserted in the theorem.
To bound the cost, note that the $\tilde x_e/4$ values are just the averages of the $\LP(C)$ values for all rounds in which the edge $e$ does not have cost $0$.
In other words,
$\sum_e c_e \tilde x_e
  \leq \frac{4}{t} \sum_{i=1}^{t} \sum_{C \in P_i} \LP^*(C)
  \leq 4 \LP^*$,
where the final inequality is from Lemma~\ref{lem:decomposed_cost}.
So it just remains to show that the $\tilde x_e$'s form a feasible solution
to LP~\eqref{LP:spanner}.

To prove this, consider an edge $e = (u,v)$,
and let $\mathcal I'_e \subseteq \mathcal I_e$ be the set of iterations
$i$ in which $N(u) \cup \{u\}$ is all in the same cluster of $P_i$
(where we fix $u$ as one of the endpoints of $e$ in an arbitrary manner).
By the second property of padded decompositions,
the probability that $N(u) \cup \{u\}$ is all in the same cluster is at least $1/2$.
The iterations are independent, so a straightforward Chernoff bound implies
that $\Pr[|\mathcal I'_e| \geq t/4] \geq 1-1/n^3$.
For a path $P \in \mathcal P_{u,v}$, set
$\tilde f_P = \frac{1}{|\mathcal I'_{(u,v)}|} \sum_{i \in \mathcal I'_{(u,v)}} f^{P_i(u,v), i}_P$.
In other words, the flow along a path from $u$ to $v$ is equal to the average
flow along it in the LP solutions that were computed in iterations when
$N(u) \cup \{u\}$ were all in the same cluster.

The capacity constraints are obviously satisfied, since each iteration
satisfies the capacity constraints, and the edge capacities are scaled
by $4/t$ while flows are scaled by a factor that can be only smaller.
Note that here we depend on the fact that \emph{all} of $N(u)$ is in
the same cluster as $u$; if some vertex $z \in N(u)$ were in a
different cluster, then in the LP solution for the cluster containing
$u$ and $v$ there could be flow sent from $u$ to $v$ through $z$.
This flow would not have the corresponding capacity added to the
$\tilde x_e$ variables, which would be a problem.

Similarly, consider the knapsack-cover constraint for some $(u,v)\in E$ and some $W \subseteq \mathcal P_{u,v}$ with $|W| \leq r$.  Then since we could send enough flow in each iteration in $\mathcal I'_{(u,v)}$, when we take the average we can still send enough flow, i.e.
\begin{align*}
  (r+1 - |W|) \tilde x_{u,v} + \sum_{P \in \mathcal P_{u,v} \setminus W} \tilde f_P
  & \geq \sum_{i \in \mathcal I'_{(u,v)}} \Big(\tfrac{4}{t} (r+1-|W|) x^i_{u,v} + \sum_{P \in \mathcal P_{u,v} \setminus W} \tfrac{1}{|\mathcal I'_{(u,v)}|} f^{P_i(u,v), i}_P  \Big) \\
  & \geq \frac{1}{|\mathcal I'_{(u,v)}|} \sum_{i \in \mathcal I'_{(u,v)}} \Big((r+1-|W|) x^i_{u,v} + \sum_{P \in \mathcal P_{u,v} \setminus W} f^{P_i(u,v), i}_P  \Big) \\
%  & \geq \frac{1}{|\mathcal I'_{(u,v)}|} \Big( |\mathcal I'_{(u,v)}| \cdot (r+1 - |W|) \Big) \\
  & = r+1 - |W|,
\end{align*}
where the last inequality is by the knapsack-cover constraint for the
cluster $P_i(u,v)$.
Thus we have a valid LP solution, completing the proof.
\end{proof}

\paragraph{Remark:}
While for our purposes it was enough to solve the LP to within a constant
factor (since we lose an $O(\log n)$ factor in the rounding anyway),
it is easy to see that we could in fact solve the LP to within
a $(1+\epsilon)$ factor.
First, we could change the padded decomposition to have the padding property
($u$ and $N(u)$ are all in the same cluster) to hold with probability
at least $1-\epsilon$,
which would require increasing the diameter of the clusters,
and thus the number of rounds it takes to solve the LP,
by an $O(1/\epsilon)$ factor.
Second, when we apply the Chernoff bound, instead of asking
the number of times the padding event occurs to be at least $t/4$,
we could ask that it is at least $(1-\epsilon)^2 t$.
By increasing $t$ by an $O(1/\epsilon^2)$ factor,
we still get the right probabilities.
Overall, the number of rounds would now be $O(\veps^{-3} \log n)$.

\section{Conclusions and Future Work}

This paper considers the problem of constructing
$r$-fault tolerant spanners and gives two basic constructions.
For general stretch bounds $k \geq 3$,
we show how to construct $r$-fault tolerant $k$-spanners whose size
is at most polynomially (in $r$) larger than spanners without fault tolerance,
improving over the previous exponential dependency (on $r$) of~\cite{CLPR09}.
Our main technique is oversampling failure sets,
in order to handle many of them in one iteration.  An interesting open question is to provide lower bounds on the size of the best $r$-fault tolerant $k$-spanner; to the best of our knowledge, no such bounds are known other than those that apply even when $r=0$.

For $k=2$ and unit edge lengths we design an $O(\log n)$-approximation
algorithm (for all $r$),
improving over the previous $O(r \log n)$ factor of~\cite{DK10}
and showing that the approximation ratio could be independent of
the desired amount of fault tolerance $r$.
Our main technique here is to design a new linear programming relaxation that
includes the exponentially many knapsack-cover inequalities of~\cite{CFLP00}.  We also provided a distributed version of the algorithm, and showed that when all edge costs are $1$ the approximation can be improved to $O(\log \Delta)$.  An interesting open question is to improve this ratio to $O(\log |E| /|V|)$, which would match the approximation known for the non-fault tolerant version.

\bibliographystyle{alphainit}
\bibliography{robi,spanner}

\end{document}